\documentclass[11pt,draftcls,onecolumn]{IEEEtran}

\usepackage{enumerate}

\usepackage{amsbsy}
\usepackage{amssymb}
\usepackage{amscd}
\usepackage{amsmath}
\usepackage{graphics}
\usepackage{graphicx}
\usepackage{epstopdf}
\usepackage{psfrag}
\usepackage{amsthm}
\usepackage{verbatim}
\usepackage{tikz}
\usetikzlibrary{shapes,arrows}
\usepackage{caption}
\def\bPhi{{\mathbf{\Phi}}}

\def\bOmega{{\mathbf{\Omega}}}
\def\bSig{{\mathbf{\Sigma}}}
\def\bPsi{{\mathbf{\Psi}}}

\def\bg{{\mathbf{g}}}
\def\bh{{\mathbf{h}}}

\def\bs{{\mathbf{s}}}

\def\bu{{\mathbf{u}}}

\def\bx{{\mathbf{x}}}
\def\by{{\mathbf{y}}}
\def\bz{{\mathbf{z}}}
\def\bA{{\mathbf{A}}}
\def\bB{{\mathbf{B}}}
\def\bC{{\mathbf{C}}}
\def\bD{{\mathbf{D}}}

\def\bG{{\mathbf{G}}}
\def\bH{{\mathbf{H}}}
\def\bI{{\mathbf{I}}}
\def\bJ{{\mathbf{J}}}
\def\bK{{\mathbf{K}}}

\def\bM{{\mathbf{M}}}

\def\bP{{\mathbf{P}}}
\def\bQ{{\mathbf{Q}}}

\def\bS{{\mathbf{S}}}

\def\bU{{\mathbf{U}}}

\def\bW{{\mathbf{W}}}
\def\bX{{\mathbf{X}}}
\def\bY{{\mathbf{Y}}}

\def\b0{{\boldsymbol{0}}}

\newcommand{\norm}[1]{{ \left\Vert #1 \right\Vert }}

\newlength{\figwidth}
\newlength{\figwidthb}
\theoremstyle{definition}
\newtheorem{thm}{Theorem}
\newtheorem{lem}{Lemma}
\newtheorem{rem}{Remark}

\newtheorem{cor}{Corollary}

\ifCLASSINFOpdf
\else
\fi

\begin{document}
\DeclareGraphicsExtensions{.eps}

\title{Full Rank Solutions for the MIMO Gaussian Wiretap Channel with an Average Power Constraint\thanks{The authors are with the Dept. of Electrical Engineering and Computer Science, University of California, Irvine, CA 92697-2625, USA. e-mail:\{afakoori, swindle\}@uci.edu}\thanks{This work was supported by the U.S. Army Research Office under the Multi-University Research
Initiative (MURI) grant W911NF-07-1-0318, and by the National Science Foundation under grant
CCF-1117983.}}

\author{\normalsize S. Ali. A. Fakoorian, {\it Student Member, IEEE} and A. Lee Swindlehurst, {\it Fellow, IEEE}}

\maketitle

\begin{abstract}
This paper considers a multiple-input multiple-output (MIMO) Gaussian
wiretap channel model, where there exists a transmitter, a legitimate
receiver and an eavesdropper, each equipped with multiple antennas. 
In this paper, we first revisit the rank property of the optimal input
covariance matrix that achieves the secrecy capacity of the multiple
antenna MIMO Gaussian wiretap channel under the average power
constraint. Next, we obtain necessary and sufficient conditions on the
MIMO wiretap channel parameters such that the optimal input covariance
matrix is full-rank, and we fully characterize the resulting
covariance matrix as well.  Numerical results are presented to
illustrate the proposed theoretical findings.

\end{abstract}

\begin{center} \bfseries EDICS: WIN-PHYL, WIN-INFO, MSP-CAPC, WIN-CONT\end{center}

\newpage

\section{Introduction}

The broadcast nature of a wireless medium makes it very
susceptible to eavesdropping, where the transmitted message is decoded by
unintended receiver(s). Recent information-theoretic research on
secure communication has focused on enhancing security at the physical
layer. The wiretap channel, first introduced and studied by Wyner \cite{Wyner},
is the most basic physical layer model that captures the problem of
communication security. Wyner showed that when an eavesdropper's
channel is a degraded version of the main channel, the source and
destination can achieve a positive secrecy rate, while ensuring that
the eavesdropper gets zero bits of information. The maximum secrecy
rate from the source to the destination is defined as the secrecy
capacity. The Gaussian wiretap channel, in which the outputs at the
legitimate receiver and at the eavesdropper are corrupted by additive
white Gaussian noise, was studied in \cite{Hellman}.

Determining the secrecy capacity of a Gaussian wiretap channel is in
general a difficult non-convex optimization problem, and has been
addressed independently in \cite{Hassibi}-\cite{Bustin}.  Oggier and
Hassibi \cite{Hassibi} and Khisti and Wornell \cite{KhistiMIMO}
followed an indirect approach using a Sato-like argument and matrix
analysis tools. They considered the problem of finding the secrecy
capacity of the Gaussian MIMO wiretap channel under the average total
power constraint, and a closed-form expression for the secrecy
capacity in the high signal-to-noise-ratio (SNR) regime was obtained
in \cite{KhistiMIMO}. In \cite{Petropulu}, the rank property of the
optimal input covariance matrix for the secrecy rate maximization
problem is discussed but the authors were unable to characterize the
solution for the general case. For some special cases of the MIMO
wiretap channel, where the solution has rank one, the optimal input
covariance matrix that achieves the secrecy capacity under the average
total power constraint was obtained in
\cite{Petropulu}-\cite{UlukusMAC}.

In \cite{Liu09}, Liu and Shamai propose a more information-theoretic approach
using the enhancement concept, originally presented by Weingarten et
al. \cite{Weingarten}, as a tool for the characterization of the MIMO Gaussian
broadcast channel capacity. Liu and Shamai have shown that an enhanced
degraded version of the channel attains the same secrecy
capacity as does a Gaussian input distribution. From the mathematical
solution in \cite{Liu09} it was evident that such an enhanced channel exists;
however it was not clear how to construct such a channel until the
work of \cite{Bustin}, which provided a closed-form expression for the secrecy
capacity under a \emph{covariance matrix} power constraint.  While
this result is interesting since the expression for the secrecy
capacity is valid for all SNR scenarios, there still exists no
computable secrecy capacity expression for the MIMO Gaussian wiretap
channel under an average total power constraint.

In this paper, we first investigate the rank property of the optimal
input covariance matrix that achieves the secrecy capacity of the
general Gaussian multiple-input multiple-output (MIMO) wiretap channel
under the average total power constraint, where the number of antennas
is arbitrary for both the transmitter and the two receivers.  Next, we
obtain the optimal input covariance matrix for the case that this
optimal covariance matrix is full-rank. Necessary and sufficient
conditions to have a full-rank optimal input covariance matrix are
characterized as well.

The rest of this paper is organized as follows. In the next section,
we describe the assumed mathematical model and revisit the current
solution for the wiretap channel under the matrix power constraint.
The rank property of the optimal input covariance matrix under the
average power constraint is investigated in Section III, and in
Section IV we characterize the conditions under which the input covariance
matrix that achieves the secrecy capacity of a wiretap channel under the
average power constraint is full-rank. In Section V, we discuss some
interesting facts regarding the optimal solution, and in Section~VI we
present numerical results to illustrate the proposed
solutions. Finally, Section~VII concludes the paper.

\textbf{Notation:} Vector-valued random variables are written with
non-boldface uppercase letters ({\em e.g.,} $X$), while the
corresponding non-boldface lowercase letter ($\bx$) denotes a specific
realization of the random variable. Scalar variables are written with
non-boldface (lowercase or uppercase) letters.  The Hermtian (i.e., 
conjugate) transpose is denoted by $(.)^H$, the matrix trace by Tr(.),
and \textbf{I} indicates an identity matrix.  Inequality $\bA \preceq \bB$ 
means that $\bA-\bB$ is Hermitian positive semi-definite. The Euclidean norm 
of the vector $\bx$ is written as $\norm{\bx}$. Mutual information between
the random variables $A$ and $B$ is denoted by $I(A;B)$, $E$
is the expectation operator, and
$\mathcal{CN}(0,\sigma^2)$ represents the complex circularly symmetric
Gaussian distribution with zero mean and variance $\sigma^2$.

\section{System Model and Prior Works}  \label{sec:ach1}
We begin with a multiple-antenna wiretap channel with $n_t$ transmit
antennas and $n_r$ and $n_e$ receive antennas at the legitimate
recipient and the eavesdropper, respectively:
\begin{align}\label{wirtp1}
\begin{split}
\by_r&=\bH \bx+ \bz_r\\
\by_e&=\bG\bx+ \bz_e
\end{split}
\end{align}
where $\bx$ is a zero-mean $n_t \times 1$ transmitted signal vector,
$\bz_r\in\mathbb{C}^{n_r\times1}$ and
$\bz_e\in\mathbb{C}^{n_e\times1}$ are additive white Gaussian noise
vectors at the receiver and eavesdropper, respectively, with
i.i.d. entries distributed as $\mathcal{CN}(0, 1)$.  The matrices
$\bH\in\mathbb{C}^{n_r\times{n_t}}$ and $\bG\in\mathbb{C}^{n_e\times
n_t}$ represent the channels associated with the receiver and the
eavesdropper, respectively.  Similar to other papers considering the
perfect secrecy rate of the wiretap channel, we assume that the
transmitter has perfect channel state information (CSI) for both the
legitimate receiver and the eavesdropper.  For the Gaussian channel,
where Gaussian inputs are an optimal choice, the secrecy capacity is
given by \cite{Hassibi}
\begin{align}\label{wirtp2}
\mathcal{C}_{sec}&=\max_{\bx} [I(X; Y_r)- I(X; Y_e)]= \max_{\bQ\succeq\b0} R(\bQ) 
\end{align}
where $R(\bQ)=\log|\bH\bQ\bH^H+\bI|-\log|\bG\bQ\bG^H+\bI|$, and
$\bQ=E\{\bx \bx^H\}$ is the input covariance matrix.

In \cite{Bustin}, the above secret communication problem was analyzed
under the matrix power-covariance constraint, defined as
\begin{align}\label{wirtp3}
\bQ\preceq \bS
\end{align}
where $\bS$ is a positive semi-definite matrix.  An explicit
expression for the secrecy capacity under~(\ref{wirtp3}) was 
obtained via applying the generalized eigenvalue decomposition to the following two positive definite matrices 
\begin{align}\label{wirtp4}
(\bS^{\frac{1}{2}}\bH^H\bH\bS^{\frac{1}{2}}+\textbf{I}\quad,\quad \bS^{\frac{1}{2}}\bG^H\bG\bS^{\frac{1}{2}}+\textbf{I})
\end{align}
In particular, there exists an invertible generalized eigenvector matrix $\bC$ such that \cite{Horn}
\begin{align}\label{wirtp5}
\bC^H\left[\bS^{\frac{1}{2}}\bG^H\bG\bS^{\frac{1}{2}}+\textbf{I}\right]\bC=\textbf{I}
\end{align}
\begin{align}\label{wirtp6}
\bC^H\left[\bS^{\frac{1}{2}}\bH^H\bH\bS^{\frac{1}{2}}+\textbf{I}\right]\bC=\mathbf{\Lambda}
\end{align}
where $\mathbf{\Lambda}=\text{diag}\{\lambda_1,...,\lambda_{n_t}\}$ is
a positive definite diagonal matrix and $\lambda_1,...,\lambda_{n_t}$
represent the generalized eigenvalues.  Without loss of generality,
we assume the eigenvalues are ordered as
$$\lambda_1\geq...\geq\lambda_b >1 \geq \lambda_{b+1}\geq...\geq\lambda_{n_t}>0$$
so that a total of $b$ $(0\leq b \leq n_t)$ are greater than 1.
Hence, we can write $\mathbf{\Lambda}$ as
\begin{align}\label{wirtp7}
\mathbf{\Lambda}=\left[
\begin{array}{ccc}
\mathbf{\Lambda}_1 & \b0\\
\b0 & \mathbf{\Lambda}_2
\end{array}
\right]
\end{align}
where $\mathbf{\Lambda}_1=\text{diag}\{\lambda_1,...,\lambda_b\}$ and
$\mathbf{\Lambda}_2=\text{diag}\{\lambda_{b+1},...,\lambda_{n_t}\}$. We can partition $\bC$ similarly:
\begin{align}\label{wirtp8}
\bC=\left[\bC_1\quad \bC_2\right]
\end{align}
where $\bC_1$ is the $n_t \times b$ submatrix representing the
generalized eigenvectors corresponding to
$\{\lambda_1,...,\lambda_b\}$ and $\bC_2$ is the $n_t \times (n_t-b)$
submatrix representing the generalized eigenvectors corresponding to
$\{\lambda_{b+1},...,\lambda_{n_t}\}$.  Using the above notation, the
secrecy capacity of the MIMO wiretap channel under the \emph{matrix}
power constraint (\ref{wirtp3}) can be expressed as
\cite{Bustin}, \cite[Theorem 3]{BCMIMO}:

\begin{cor}\label{cor1}
Under the matrix power constraint (\ref{wirtp3}), the secrecy capacity of the MIMO Gaussian wiretap channel is given by
\begin{align}\label{wirtp9}
\mathcal{C}_{sec}(\bS)=\sum_{i=1}^{b}\log\lambda_i=\log|\mathbf{\Lambda}_1|
\end{align}
where the optimal input covariance matrix $\bQ_S^*$ that 
maximizes~(\ref{wirtp2}) and attains~(\ref{wirtp9}) is given by 
\begin{align}\label{wirtp10}
\bQ_S^*= \bS^{\frac{1}{2}}\bC \left[
\begin{array}{ccc}
(\bC_1^H\bC_1)^{-1} & \b0\\
\b0 & \b0
\end{array}
\right]\bC^H \bS^{\frac{1}{2}} \; .
\end{align}
\end{cor}

\begin{rem}\label{rem1}
From (\ref{wirtp5}) and (\ref{wirtp6}), one can easily confirm that if
$\bH^H\bH \preceq \bG^H\bG$, then for any $\bS \succeq \b0$ we have
$\mathbf{\Lambda} \preceq \bI$. In other words, in this case the pencil in
(\ref{wirtp4}) has no generalized eigenvalue bigger than 1. Thus,
$\mathcal{C}_{sec}(\bS)=0$ for any $\bS \succeq \b0$.
\end{rem} 

In this paper, we consider the secrecy capacity problem in
(\ref{wirtp2}) under the \emph{average} power constraint:
\begin{align}\label{wirtp11}
\mathrm{Tr}(E\{\bx \bx^H\})= \mathrm{Tr}(\bQ)\leq P_t.
\end{align}
For this constraint, no computable secrecy capacity expression has
been derived to date for the general MIMO case.  In principle, one
would have to find the secrecy capacity through an exhaustive search
over the set $\{\bS: \bS\succeq \b0, \text{Tr}(\bS)\leq P\}$
\cite[Lemma 1]{Weingarten}, \cite{BCMIMO}:
\begin{align}\label{wirtp12}
\mathcal{C}_{sec}(P_t)=\max_{\bS\succeq \b0, \text{Tr}(\bS) = P_t}\mathcal{C}_{sec}(\bS) \; .
\end{align}
where for any given semidefinite $\bS$, $\mathcal{C}_{sec}(\bS)$
should be computed as given by~(\ref{wirtp9}).

In the next section, we investigate the rank of the optimal input
covariance matrix $\bQ^*$ that attains
$\mathcal{C}_{sec}(P_t)$. Next, in Section IV, we obtain 
the optimal $\bQ^*$ under the average power
constraint for the case that $\bQ^*$ is full-rank.

\section{Rank Property of the Optimal Solution under an Average Power Constraint}  \label{sec:ach1}

First, we note that the problem under $\mathrm{Tr}(\bQ)\leq P_t$ is
equivalent to that under $\mathrm{Tr}(\bQ)=P_t$
\cite{Hassibi,Petropulu}\footnote{For this statement, and also for the
following results in the paper, we exclude the special case $\bH^H\bH
\preceq \bG^H\bG$ for which the $C_{sec}$ is trivially $0$ for any
$\bS \succeq \b0$, and consequently for any $P_t$, as pointed out in
Remark \ref{rem1}.}. Also note that in (\ref{wirtp12}), this implies
that we have $\mathrm{Tr}(\bS)=P_t$ instead of $\mathrm{Tr}(\bS)\le
P_t$.

We are interested in finding the optimal $\widehat{\bS}$ which
maximizes the problem (\ref{wirtp12}). Let us assume that we have
found the optimal $\widehat{\bS}$. Consequently, from (\ref{wirtp10}),
the optimal input covariance matrix that attains
$\mathcal{C}_{sec}(P_t)$ is given by
\begin{align}\label{wirtp13}
\bQ^*= \widehat{\bS}^{\frac{1}{2}}\widehat{\bC} \left[
\begin{array}{ccc}
(\widehat{\bC}_1^H\widehat{\bC}_1)^{-1} & \b0\\
\b0 & \b0
\end{array}
\right]\widehat{\bC}^H \widehat{\bS}^{\frac{1}{2}} \; ,
\end{align}
where $\widehat{\bC}$ and $\widehat{\bC}_1$ have respectively the same
definitions as those of $\bC$ and $\bC_1$, given by
(\ref{wirtp5})-(\ref{wirtp8}), but here for the pencil
$(\widehat{\bS}^{\frac{1}{2}}\bH^H\bH\widehat{\bS}^{\frac{1}{2}}+\bI\;,\;\widehat{\bS}^{\frac{1}{2}}\bG^H\bG\widehat{\bS}^{\frac{1}{2}}+\bI)$. Note
that $\bQ^*$ can be rewritten as
\begin{align}\label{wirtp14}
\bQ^*&= \widehat{\bS}^{\frac{1}{2}}\left[\widehat{\bC}_1\quad\widehat{\bC}_2\right] 
\left[
\begin{array}{ccc}
(\widehat{\bC}_1^H\widehat{\bC}_1)^{-1} & \b0\\
\b0 & \b0
\end{array}
\right]
\left[
\begin{array}{ccc}
\widehat{\bC}^H_1\\
\widehat{\bC}^H_2
\end{array}
\right]
\widehat{\bS}^{\frac{1}{2}} \nonumber\\
&=\widehat{\bS}^{\frac{1}{2}}\,\widehat{\bC}_1 (\widehat{\bC}_1^H\widehat{\bC}_1)^{-1} \widehat{\bC}^H_1\, \widehat{\bS}^{\frac{1}{2}} \nonumber\\
&=\widehat{\bS}^{\frac{1}{2}}\,\bP_{\widehat{\bC}_1}\,\widehat{\bS}^{\frac{1}{2}}
\end{align}
where $\bP_{\widehat{\bC}_1}=\widehat{\bC}_1
(\widehat{\bC}_1^H\widehat{\bC}_1)^{-1} \widehat{\bC}^H_1$ is the
projection matrix onto the space of $\widehat{\bC}_1$. Moreover, let
$\bP^\perp_{\widehat{\bC}_1}=\bI-\bP_{\widehat{\bC}_1}$ be the
projection onto the space orthogonal to $\widehat{\bC}_1$. We have
\begin{align}
\mathrm{Tr}(\bQ^*) &= \mathrm{Tr}(\widehat{\bS}^{\frac{1}{2}}\,\bP_{\widehat{\bC}_1}\,\widehat{\bS}^{\frac{1}{2}}) \nonumber\\
&= \mathrm{Tr}(\widehat{\bS}\,\bP_{\widehat{\bC}_1})   \label{wirtp15}\\
&= \mathrm{Tr}(\widehat{\bS}\,\bP_{\widehat{\bC}_1} \bP_{\widehat{\bC}_1})   \label{wirtp16}\\
&= \mathrm{Tr}(\bP_{\widehat{\bC}_1}\widehat{\bS}\,\bP_{\widehat{\bC}_1} )   \label{wirtp17}
\end{align}
where (\ref{wirtp15}) comes from the fact that
$\mathrm{Tr}(\bA\bB)=\mathrm{Tr}(\bB\bA)$, and (\ref{wirtp16}) is
because $\bP_{\widehat{\bC}_1}=\bP_{\widehat{\bC}_1}
\bP_{\widehat{\bC}_1}$. Similarly we have
\begin{align}
\mathrm{Tr}(\widehat{\bS}) &= \mathrm{Tr}\left((\bP_{\widehat{\bC}_1}+\bP^\perp_{\widehat{\bC}_1})\,\widehat{\bS}\,(\bP_{\widehat{\bC}_1}+\bP^\perp_{\widehat{\bC}_1})\right) \nonumber\\
&= \mathrm{Tr}(\bP_{\widehat{\bC}_1}\,\widehat{\bS}\,\bP_{\widehat{\bC}_1}) +\mathrm{Tr}(\bP^\perp_{\widehat{\bC}_1}\,\widehat{\bS}\,\bP^\perp_{\widehat{\bC}_1})   \label{wirtp18}\\
&= \mathrm{Tr}(\bQ_{\widehat{S}}^*) + \mathrm{Tr}(\bP^\perp_{\widehat{\bC}_1}\,\widehat{\bS}\,\bP^\perp_{\widehat{\bC}_1})  \label{wirtp19}
\end{align}
where (\ref{wirtp19}) results from (\ref{wirtp17}).

\begin{lem} \label{lem1}
For the optimal $\widehat{\bS}$, we have $\mathrm{span}\{\widehat{\bC}_1\}= \mathrm{span}\{\widehat{\bS}\}.$
\end{lem}
\begin{proof}: The proof is obtained using (\ref{wirtp19}), and by noting that for the optimal $\widehat{\bS}$ we must have $\mathrm{Tr}(\widehat{\bS})=
\mathrm{Tr}(\bQ^*)=P_t$. This means that we must have $\mathrm{Tr}(\bP^\perp_{\widehat{\bC}_1}\,\widehat{\bS}\,\bP^\perp_{\widehat{\bC}_1})=0$, or equivalently $\bP^\perp_{\widehat{\bC}_1}\,\widehat{\bS} = \b0$, which completes the proof. 
\end{proof}

Using Lemma \ref{lem1} in (\ref{wirtp14}) we have
\begin{align} \label{wirtp20}
\bQ^*=\widehat{\bS}.
\end{align}
The following lemma reveals another property of the optimal input covariance matrix under the average power constraint. 

\begin{lem}\label{lem2}
For the optimal $\widehat{\bS}$, i.e. $\bQ^*$, the pencil $(\widehat{\bS}^{\frac{1}{2}}\bH^H\bH\widehat{\bS}^{\frac{1}{2}}+\bI\;,\;\widehat{\bS}^{\frac{1}{2}}\bG^H\bG\widehat{\bS}^{\frac{1}{2}}+\bI)$ has no generalized eigenvalue less than one:
\begin{align}
\widehat{\bC}^H\left[\widehat{\bS}^{\frac{1}{2}}\bH^H\bH\widehat{\bS}^{\frac{1}{2}}+\bI\right]\widehat{\bC}&=\left[
\begin{array}{ccc}
\widehat{\mathbf{\Lambda}}_1 & \b0\\
\b0 & \bI
\end{array}
\right] \label{wirtp21} \\\label{wirtp22}
\widehat{\bC}^H\left[\widehat{\bS}^{\frac{1}{2}}\bG^H\bG\widehat{\bS}^{\frac{1}{2}}+\bI\right]\widehat{\bC}&=\left[
\begin{array}{ccc}
\bI & \b0\\
\b0 & \bI
\end{array}
\right]
\end{align}
where $\widehat{\mathbf{\Lambda}}_2=\bI$ corresponds to the
generalized eigenvalues equal to one.
\end{lem}
\begin{proof}
We note that any vector which lies in the null space of
$\widehat{\bS}$ can be a generalized eigenvector of the pencil
$(\widehat{\bS}^{\frac{1}{2}}\bH^H\bH\widehat{\bS}^{\frac{1}{2}}+\bI\;,\;\widehat{\bS}^{\frac{1}{2}}\bG^H\bG\widehat{\bS}^{\frac{1}{2}}+\bI)$,
with a generalized eigenvalue equal to 1. Such vectors span the 
null space of $\widehat{\bS}$, i.e.,
$\mathrm{span}\{\widehat{\bC}_2\}=\mathrm{span}\{\widehat{\bS}\}^\perp$. On
the other hand, from Lemma \ref{lem1},
$\mathrm{span}\{\widehat{\bC}_1\}=
\mathrm{span}\{\widehat{\bS}\}$. Thus for the optimal $\widehat{\bS}$,
all generalized eigenvectors of the pencil
$(\widehat{\bS}^{\frac{1}{2}}\bH^H\bH\widehat{\bS}^{\frac{1}{2}}+\bI\;,\;\widehat{\bS}^{\frac{1}{2}}\bG^H\bG\widehat{\bS}^{\frac{1}{2}}+\bI)$
correspond to generalized eigenvalues either bigger than 
or equal to 1.
\end{proof}

Let $b$ denote number of generalized eigenvalues of the pencil
$(\widehat{\bS}^{\frac{1}{2}}\bH^H\bH\widehat{\bS}^{\frac{1}{2}}+\bI\;,\;\widehat{\bS}^{\frac{1}{2}}\bG^H\bG\widehat{\bS}^{\frac{1}{2}}+\bI)$
that are strictly bigger than 1, where again $\widehat{\bS}=\bQ^*$
represents the optimal input covariance matrix that attains the
secrecy capacity under the average power constraint given by
(\ref{wirtp11}). From Lemma \ref{lem1}, we have
$\mathrm{rank}(\bQ^*)=\mathrm{rank}(\widehat{\bC}_1)=b$.

\begin{thm}\label{thm1}
For the optimal $\bQ^*$ we have
\begin{align} \label{wirtp23}
\mathrm{rank}(\bQ^*)\le m
\end{align}
where $m$ is the number of positive eigenvalues of the matrix 
$\bH^H\bH-\bG^H\bG$.
\end{thm}
\begin{proof}
Subtracting (\ref{wirtp22}) from (\ref{wirtp21}), a 
straightforward computation yields
\begin{align} \label{wirtp24}
\widehat{\bS}^{\frac{1}{2}}\left[\bH^H\bH-\bG^H\bG\right]\widehat{\bS}^{\frac{1}{2}}=\widehat{\bC}^{-H}\left[
\begin{array}{ccc}
\widehat{\mathbf{\Lambda}}_1-\bI & \b0\\
\b0 & \b0
\end{array}
\right] \widehat{\bC}^{-1}\succeq\b0.
\end{align}
From (\ref{wirtp24}), we note that
$\widehat{\bS}^{\frac{1}{2}}\left[\bH^H\bH-\bG^H\bG\right]\widehat{\bS}^{\frac{1}{2}}\succeq\b0$,
from which it follows that $\mathrm{rank}(\bQ^*)\le m$.
\end{proof}

\begin{rem}\label{rem2} 
From Theorem \ref{thm1}, one can easily confirm that the optimal
$\bQ^*$ can be full rank only in the case that $m=n_t$,
i.e. $\bH^H\bH\succ \bG^H\bG$. For all other scenarios, the optimal
$\bQ^*$ will be low rank.  The authors of \cite{Hassibi,Petropulu} use
the Karush-Kuhn-Tucker (KKT) conditions on problem (\ref{wirtp2}) to
make a similar statement, but they do not show what the rank of the
optimal $\bQ^*$ will be.
\end{rem}

The following lemma will be used for the computations in the next section.
\begin{lem}\label{lem3}
For the case of $\bH^H\bH\succ \bG^H\bG$, for any $n_t\times n_t$
matrix $\bS\succeq\b0$, all the generalized eigenvalues of the pencil
$(\bS^{\frac{1}{2}}\bH^H\bH\bS^{\frac{1}{2}}+\bI\;,\;\bS^{\frac{1}{2}}\bG^H\bG\bS^{\frac{1}{2}}+\bI)$
are strictly bigger than 1, i.e. $\mathbf{\Lambda} \succ \bI$,
\emph{iff} $\bS$ is full rank, i.e. $\bS \succ \b0$.
\end{lem}
\begin{proof}
The claim is easily proved by considering the rank of both sides of  (\ref{wirtp24}).
\end{proof}

\section{Characterization of the Optimal Full-Rank Solution}

In this section, we characterize the secrecy capacity under the
average power constraint for a particular class of MIMO Gaussian
wiretap channel where the optimal solution $\bQ^*$ is full rank. While
necessary conditions for a full-rank $\bQ^*$ were characterized in the
previous section, here we derive sufficient conditions as well.

We begin by rewriting problem (\ref{wirtp2}) here:
\begin{align}\label{wirtp25}
\mathcal{C}_{sec}(P_t)&= \max_{\bQ\succeq\b0,\; \mathrm{Tr}(\bQ)=P_t} \log|\bH\bQ\bH^H+\bI|-\log|\bG\bQ\bG^H+\bI| \; .
\end{align}
The Lagrangian associated with this problem is given by
\begin{align}\label{wirtp26}
\mathcal{L}=\log|\bH\bQ\bH^H+\bI|-\log|\bG\bQ\bG^H+\bI|-\mu (\mathrm{Tr}(\bQ)-P_t)+ \mathrm{Tr}(\bM\bQ)
\end{align}
where $\mu>0$ and $\bM\succeq \b0$ are the Lagrange multipliers. The optimal $\bQ^*$ must satisfy the following KKT conditions:
\begin{align}
\bH^H(\bH\bQ^*\bH^H+\bI)^{-1}\bH-\bG^H(\bG\bQ^*\bG^H+\bI)^{-1}\bG &=\mu\bI - \bM  \label{wirtp27}\\
\mu (\mathrm{Tr}(\bQ^*)-P_t) &=0 \label{wirtp28} \\
\bQ^*\bM= \bM\bQ^*&=\b0 \; . 
\label{wirtp29}
\end{align}
Using the matrix inversion lemma \cite{Horn}, (\ref{wirtp27}) can be written as
\begin{align}\label{wirtp30}
(\bH^H\bH\bQ^*+\bI)^{-1}\bH^H\bH-\bG^H\bG(\bQ^*\bG^H\bG+\bI)^{-1} &=\mu\bI - \bM
\; .
\end{align}
Left multiplication by $(\bH^H\bH\bQ^*+\bI)$ and right multiplication by $(\bQ^*\bG^H\bG+\bI)$ of both sides of (\ref{wirtp30}) yields
\begin{align}
\bH^H\bH-\bG^H\bG &=\mu\,(\bH^H\bH\bQ^*+\bI) \; (\bQ^*\bG^H\bG+\bI)- \bM  \label{wirtp31}\\
&= \mu\,(\bG^H\bG\bQ^*+\bI) \; (\bQ^*\bH^H\bH+\bI)- \bM  \; , \label{wirtp32}
\end{align}
where in obtaining (\ref{wirtp31}) we have used the KKT 
condition~(\ref{wirtp29}), and Eq.~(\ref{wirtp32}) comes from the fact 
that~(\ref{wirtp31}) is Hermitian.
 
We are considering problem (\ref{wirtp25}) for the case that
$\bH^H\bH-\bG^H\bG$ is strictly positive definite, i.e. $\bH^H\bH\succ
\bG^H\bG$, since this is the necessary condition for having
a full rank optimal $\bQ^*$. As we characterize the full rank $\bQ^*$,
the sufficient conditions are revealed as well.

\begin{rem}\label{rem3}
By following exactly the same steps as in the proof of
\cite[Proposition 5]{Hassibi}, one can easily show that for the case
of $\bH^H\bH\succ \bG^H\bG$, the optimization problem (\ref{wirtp25})
is convex\footnote{In fact, the optimization problem (\ref{wirtp25})
is convex in $\bQ$ when $\bH^H\bH\succeq\bG^H\bG$.} in $\bQ$.
\end{rem}

Thus for the case of interest, the KKT conditions (\ref{wirtp28}),
(\ref{wirtp29}) and (\ref{wirtp31}) are necessary and sufficient
conditions for the optimality of $\bQ^*$. In other words, any
$\bQ\succeq\b0$ that satisfies those conditions is an optimal solution
for the problem~(\ref{wirtp25}). 
By the KKT condition (\ref{wirtp29}), a full
rank $\bQ^*$ follows that $\bM=\b0$. Thus, (\ref{wirtp31}) and
(\ref{wirtp32}) simplify to
\begin{align}
\bH^H\bH-\bG^H\bG &=\mu\,(\bH^H\bH\bQ^*+\bI) \; (\bQ^*\bG^H\bG+\bI)  \label{wirtp33}\\
&= \mu\,(\bG^H\bG\bQ^*+\bI) \; (\bQ^*\bH^H\bH+\bI) \; . \label{wirtp34}
\end{align}

\begin{lem}\label{lem4}
Let the diagonal matrix $\bD$ and the unitary matrix $\bPhi_{\bar{\bs}}$ respectively denote the eigenvalue and eigenvector matrices of $(\bar{\bS}^{\frac{1}{2}}\bG^H\bG\bar{\bS}^{\frac{1}{2}}+\bI)$, where we set $\bar{\bS}=(\bH^H\bH-\bG^H\bG)^{-1}$:
\begin{align}  \label{wirtp37}
(\bar{\bS}^{\frac{1}{2}}\bG^H\bG\bar{\bS}^{\frac{1}{2}}+\bI)=\bPhi_{\bar{\bs}}\; \bD\;\bPhi_{\bar{\bs}}^H.
\end{align}
Then we have
\begin{align}
\bH^H\bH &=\bar{\bS}^{-\frac{1}{2}}\, \bPhi_{\bar{\bs}}\; \bD  
\;\bPhi_{\bar{\bs}}^H  \bar{\bS}^{-\frac{1}{2}}  \label{wirtp35}\\
\bG^H\bG &=\bar{\bS}^{-\frac{1}{2}}\, \bPhi_{\bar{\bs}} \; (\bD-\bI)\; \bPhi_{\bar{\bs}}^H   \bar{\bS}^{-\frac{1}{2}} \; . \label{wirtp36}
\end{align}
\end{lem}
\begin{proof}
Eq. (\ref{wirtp36}) comes directly from (\ref{wirtp37}). Please refer to Appendix A for details on obtaining (\ref{wirtp35}).   
\end{proof}

Using  (\ref{wirtp35}) and (\ref{wirtp36}) in (\ref{wirtp33}), after some
simplification we have
\begin{align} \label{wirtp38}
\bar{\bS}^{-\frac{1}{2}}\, \bPhi_{\bar{\bs}} \;\bPhi_{\bar{\bs}}^H  \bar{\bS}^{-\frac{1}{2}} &= \mu\,(\bar{\bS}^{-\frac{1}{2}}\, \bPhi_{\bar{\bs}}\; \bD  
\;\bPhi_{\bar{\bs}}^H  \bar{\bS}^{-\frac{1}{2}}\;\bQ^*+\bI) \; (\bQ^*\;\bar{\bS}^{-\frac{1}{2}}\, \bPhi_{\bar{\bs}} \; (\bD-\bI)\; \bPhi_{\bar{\bs}}^H   \bar{\bS}^{-\frac{1}{2}}+\bI)    \nonumber\\
&= \mu\,\bar{\bS}^{-\frac{1}{2}}\, \bPhi_{\bar{\bs}}\; (\bD  
\;\bPhi_{\bar{\bs}}^H  \bar{\bS}^{-\frac{1}{2}}\;\bQ^*+\bPhi_{\bar{\bs}}^H\bar{\bS}^{\frac{1}{2}}) \; (\bQ^*\;\bar{\bS}^{-\frac{1}{2}}\, \bPhi_{\bar{\bs}} \; (\bD-\bI)+\bar{\bS}^{\frac{1}{2}}\, \bPhi_{\bar{\bs}})\; \bPhi_{\bar{\bs}}^H   \bar{\bS}^{-\frac{1}{2}}  \nonumber\\
&= \mu\,\bar{\bS}^{-\frac{1}{2}}\, \bPhi_{\bar{\bs}}\; (\bD  \,\bW+\bI) \;\bPhi_{\bar{\bs}}^H\bar{\bS} \bPhi_{\bar{\bs}}\; \left(\bW (\bD-\bI)+\bI\right)\; \bPhi_{\bar{\bs}}^H   \bar{\bS}^{-\frac{1}{2}} \; ,
\end{align}
where in (\ref{wirtp38}) we defined
\begin{align} \label{wirtp39}
\bW=\bPhi_{\bar{\bs}}^H  \bar{\bS}^{-\frac{1}{2}}\;\bQ^*\bar{\bS}^{-\frac{1}{2}}\bPhi_{\bar{\bs}} \;.
\end{align}
From (\ref{wirtp38}), we get
\begin{align} \label{wirtp40}
\bI= \mu\, (\bD  \,\bW+\bI) \;\bPhi_{\bar{\bs}}^H\bar{\bS} \bPhi_{\bar{\bs}}\; \left(\bW (\bD-\bI)+\bI\right).
\end{align}

Let $\bX\succ\b0$ and $\bY\succ\b0$ be two diagonal matrices, which will be defined soon. Left  multiplication by $\bX$ and right multiplication by $\bY$ of both sides of (\ref{wirtp40}) gives
\begin{align} \label{wirtp41}
\bX\bY= \mu\, (\bX\bD  \,\bW+\bX) \;\bPhi_{\bar{\bs}}^H\bar{\bS} \bPhi_{\bar{\bs}}\; \left(\bW (\bD-\bI)\bY+\bY\right).
\end{align}
We find $\bX$ and $\bY$ by solving the following set of equations 
\begin{align} \label{wirtp42}
\begin{split}
\bX\bY&= \bI \\
\bX\bD &=(\bD-\bI)\bY
\end{split}
\end{align}
which results in
\begin{align} \label{wirtp43}
\bX=\bY^{-1}=\left(\bI-\bD^{-1}\right)^{\frac{1}{2}}\;.
\end{align}
Substituting (\ref{wirtp43}) into (\ref{wirtp41}), we have
\begin{align} \label{wirtp44}
\bI&= \mu\, \left(\bX\bD  \,\bW+\bX\right) \;\bPhi_{\bar{\bs}}^H\bar{\bS} \bPhi_{\bar{\bs}}\; \left(\bW\, \bD\bX+\bX^{-1}\right)  \nonumber\\
&= \mu\, \left(\bX\bD  \,\bW\, \bD\bX+\bX^2\bD\right) \;\bX^{-1}\bD^{-1}\bPhi_{\bar{\bs}}^H\bar{\bS} \bPhi_{\bar{\bs}}\bD^{-1}\bX^{-1}\; \left(\bX\bD\,\bW\, \bD\bX+\bD\right)   \nonumber\\
&= \mu\, \left(\bX\bD  \,\bW\, \bD\bX+(\bD-\bI)\right) \;\bX^{-1}\bD^{-1}\bPhi_{\bar{\bs}}^H\bar{\bS} \bPhi_{\bar{\bs}}\bD^{-1}\bX^{-1}\; \left(\bX\bD\,\bW\, \bD\bX+\bD\right) \; .
\end{align}
where we have used (\ref{wirtp43}) in obtaining (\ref{wirtp44}). 

Using an approach similar to what we used to obtain (\ref{wirtp44}) from (\ref{wirtp33}), one can show that\footnote{Clearly (\ref{wirtp45}) is also trivially obtained from (\ref{wirtp44}).} 
\begin{align} \label{wirtp45}
\bI&= \mu\, \left(\bX\bD  \,\bW\, \bD\bX+\bD\right) \;\bX^{-1}\bD^{-1}\bPhi_{\bar{\bs}}^H\bar{\bS} \bPhi_{\bar{\bs}}\bD^{-1}\bX^{-1}\; \left(\bX\bD\,\bW\, \bD\bX+(\bD-\bI)\right) \; .
\end{align}
Define $\bK$ as
\begin{align} \label{wirtp46}
\bK = \bX\bD  \,\bW\, \bD\bX = \left(\bI-\bD^{-1}\right)^{\frac{1}{2}} \bD\,\bW\, \bD\left(\bI-\bD^{-1}\right)^{\frac{1}{2}} \; ,
\end{align}
where $\bD$ and $\bW$ are respectively given by (\ref{wirtp37}) and (\ref{wirtp39}). Moreover, let 
\begin{align} \label{wirtp47}
\bar{\bSig}=\bX^{-1}\bD^{-1}\bPhi_{\bar{\bs}}^H\bar{\bS} \bPhi_{\bar{\bs}}\bD^{-1}\bX^{-1}=  \left(\bI-\bD^{-1}\right)^{-\frac{1}{2}}\bD^{-1}\bPhi_{\bar{\bs}}^H\bar{\bS} \bPhi_{\bar{\bs}}\bD^{-1} \left(\bI-\bD^{-1}\right)^{-\frac{1}{2}} \; .
\end{align}
Using (\ref{wirtp46}) and (\ref{wirtp47}) in  (\ref{wirtp44}) and (\ref{wirtp45}), we have
\begin{align} \label{wirtp48}
\begin{split}
\bI&= \mu\, \left(\bK+\bD\right) \;\bar{\bSig}\; \left(\bK+(\bD-\bI)\right) \\
&=\mu\, \left(\bK+(\bD-\bI)\right) \;\bar{\bSig}\;  \left(\bK+\bD\right) \; .
\end{split}
\end{align}
We note from (\ref{wirtp48}) that 
\begin{align} \label{wirtp49}
\mu\,\bar{\bSig}= \left(\bK+\bD\right)^{-1}\left(\bK+(\bD-\bI)\right)^{-1}= \left(\bK+(\bD-\bI)\right)^{-1} \left(\bK+\bD\right)^{-1} \;.
\end{align}
This result implies that, for the optimal $\bQ^*$, $(\bK+(\bD-\bI))^{-1}$, $(\bK+\bD)^{-1}$ and $\bar{\bSig}$ commute and have the same eigenvectors \cite{Horn}.

Denote the eigenvalue decomposition of $\bar{\bSig}$ as 
\begin{align} \label{wirtp50}
\bar{\bSig}= \bU\,\bOmega\,\bU^H\;.
\end{align}
Based on the argument made after (\ref{wirtp49}), we have
\begin{align} \label{wirtp51}
\begin{split}
\bK+(\bD-\bI)&= \bU\,\mathbf{\Lambda}_1\,\bU^H\; \\
\bK+\bD&= \bU\,\mathbf{\Lambda}_2\,\bU^H\; ,
\end{split}
\end{align}
where one can easily confirm that
$\mathbf{\Lambda}_2=\mathbf{\Lambda}_1+\bI$. By replacing
(\ref{wirtp50}) and (\ref{wirtp51}) in (\ref{wirtp48}), and noting
that $\bU^H\bU=\bU\bU^H=\bI$, and that $\mathbf{\Lambda}_1$,
$\mathbf{\Lambda}_2$ and $\bOmega$ are all diagonal, we have
\begin{align} \label{wirtp52}
\bI= \mu\,\bU \mathbf{\Lambda}_1\,\bOmega\,\mathbf{\Lambda}_2\bU^H= \mu\,\mathbf{\Lambda}_1\,\bOmega\,\mathbf{\Lambda}_2= \mu\, (\mathbf{\Lambda}^2_1+\mathbf{\Lambda}_1)\, \bOmega \; .
\end{align}
Recall that the unknown parameters in (\ref{wirtp52}) are the
diagonal matrix $\mathbf{\Lambda}_1$ and the scalar $\mu>0$. Let
$\lambda_{i1}$ and $\omega_i$ denote the $i$th diagonal element of
$\mathbf{\Lambda}_1$ and $\bOmega$, respectively. From
(\ref{wirtp52}), we can solve for $\lambda_{i1}$ and obtain
\begin{align}\label{wirtp53}
\lambda_{i1}=\frac{1}{2}\left(-1+\sqrt{1+\frac{4}{\mu \omega_i}}\right) \; ,
\end{align}
where the Lagrange multiplier $\mu>0$ is chosen to satisfy the power
constraint $\mathrm{Tr}(\bQ^*)= P_t$, as will be explained below.

\begin{thm}\label{thm2}
The optimal full-rank input covariance matrix that attains the secrecy
capacity for the average power constraint is given by
\begin{align}\label{wirtp54}
\bQ^*=\bar{\bS}^{\frac{1}{2}}\bPhi_{\bar{\bs}}\left(\bI-\bD^{-1}\right)^{-\frac{1}{2}}\bD^{-1}\left(\bU \mathbf{\Lambda}_1\bU^H +\bI-\bD\right) \bD^{-1}\left(\bI-\bD^{-1}\right)^{-\frac{1}{2}} \bPhi_{\bar{\bs}}^H  \bar{\bS}^{\frac{1}{2}}
\end{align}
\emph{iff} 
\begin{itemize}
\item $\bH^H\bH-\bG^H\bG\succ\b0$, and 
\item for the given $P_t$,  
\begin{align} \label{wirtp55}
\bU \mathbf{\Lambda}_1\bU^H \succ\bD-\bI
\end{align}
\end{itemize}
where  $\mathbf{\Lambda}_1$ is given by (\ref{wirtp53}).
\end{thm}
\begin{proof}
The proof is obtained by obtaining $\bK$ from (\ref{wirtp51}), and
substituting it back into (\ref{wirtp46}) and (\ref{wirtp39}) via
straightforward computations.
\end{proof}

To fully characterize $\bQ^*$, one must obtain the Lagrange
multiplier $\mu>0$ such that $\mathrm{Tr}(\bQ^*)= P_t$. As
(\ref{wirtp53}) shows, $\mathrm{Tr}(\bQ^*)$ is monotonically decreasing
with $\mu$:
$$\lim_{\mu\rightarrow 0}\mathrm{Tr}(\bQ^*) = \infty, \quad \quad \mathrm{and}  \quad \quad \lim_{\mu\rightarrow \infty}\mathrm{Tr}(\bQ^*) = -\mathrm{Tr}(\bar{\bS}^{\frac{1}{2}}\bPhi_{\bar{\bs}}\bD^{-1}\bPhi_{\bar{\bs}}^H  \bar{\bS}^{\frac{1}{2}})<0\;.$$
Thus for any transmit power $P_t$, there exists a Lagrange
multiplier $\mu>0$ for which $\mathrm{Tr}(\bQ^*)= P_t$. The appropriate
value of $\mu$ can be easily found using, for example, the bisection method.

Note that Theorem \ref{thm2} also reveals the necessary and sufficient
conditions for having a full rank optimal $\bQ^*$. While for any
transmit power $P_t$ one can find a Lagrange multiplier $\mu>0$ for
which $\mathrm{Tr}(\bQ^*)= P_t$ and $\bU \mathbf{\Lambda}_1\bU^H
\succ\b0$, to have a full rank $\bQ^*$, (\ref{wirtp55}) must be
satisfied. Also recall from~(\ref{wrtap9}) that $\bD-\bI\succeq\b0$. The
flowchart in Fig.~\ref{fig1} summarizes the steps required to
calculate the optimal full-rank $\bQ^*$ for the case of
$\bH^H\bH-\bG^H\bG\succ\b0$.


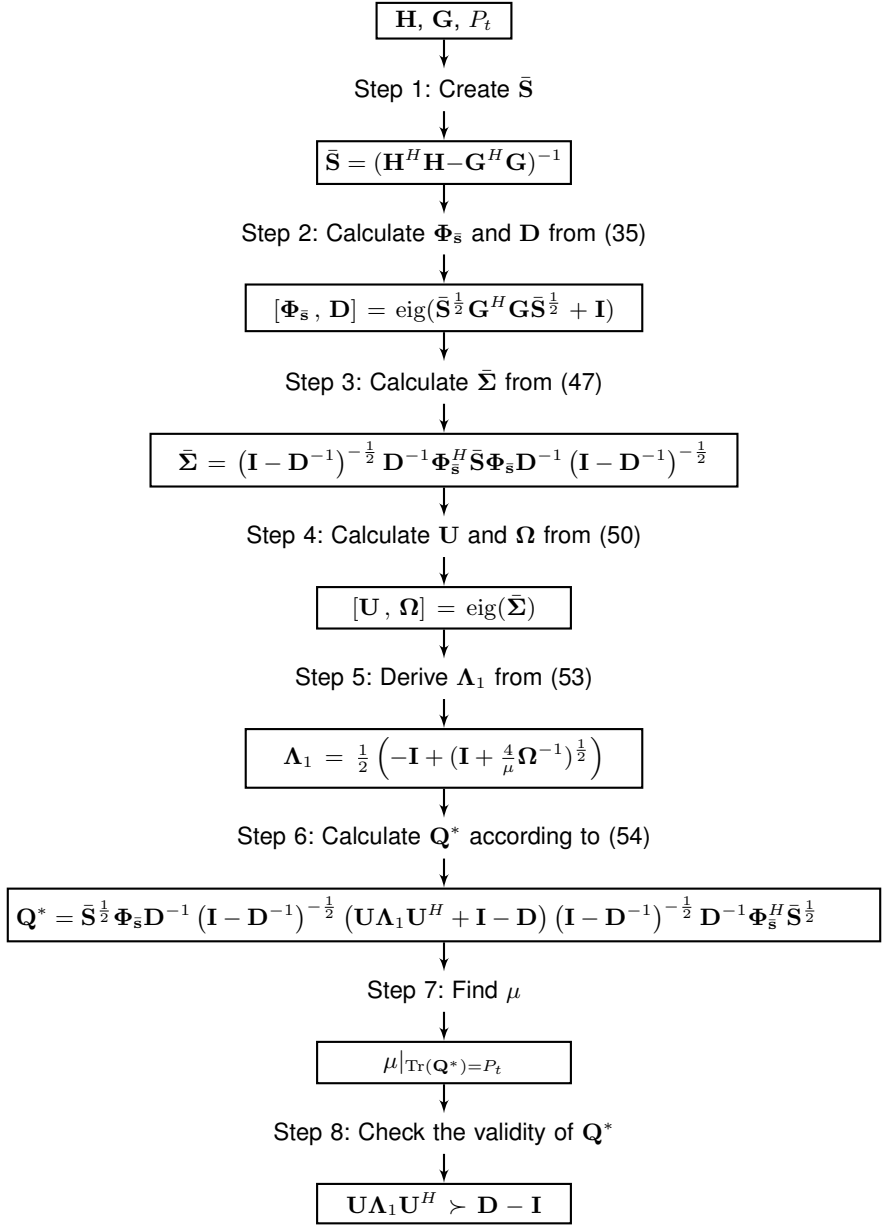
\begin{figure}[t]
\begin{center}
  \sffamily\footnotesize
  \begin{tikzpicture}[auto,
    decision/.style={diamond, draw=black, thick, fill=white, 
    text width=10em, text badly centered, minimum height=1pt,  inner sep=1pt},
    block_center1/.style ={rectangle, draw=black, thick, fill=white,
      text width=5em, text centered, minimum height=1em},
    block_center2/.style ={rectangle, draw=black, thick, fill=white,
      text width=10em, text centered, minimum height=1em},
    block_center3/.style ={rectangle, draw=black, thick, fill=white,
      text width=16em, text centered, minimum height=1em},
 block_center4/.style ={rectangle, draw=black, thick, fill=white,
      text width=24em, text centered, minimum height=1em},
    block_noborder/.style ={rectangle, draw=none, thick, fill=none,
      text width=18em, text centered, minimum height=1em},
    block_lost/.style ={rectangle, draw=black, thick, fill=white,
      text width=36em, text ragged, minimum height=1em},
      line/.style ={draw, thick, -latex', shorten >=0pt}]
    \matrix [column sep=5mm,row sep=4mm] {
      \node [block_center1] (n1) {$\bH$, $\bG$, $P_t$};\\
      \node [block_noborder] (s1) {Step 1: Create $\bar{\bS}$}; \\
      \node [block_center2] (n2) {$\bar{\bS}=(\bH^H\bH-\bG^H\bG)^{-1}$};\\
      \node [block_noborder] (s2) {Step 2: Calculate $\bPhi_{\bar{\bs}}$ and $\bD$ from (\ref{wirtp37})}; \\
      \node [block_center3] (n3) {$[\bPhi_{\bar{\bs}}\,,\,\bD]=\mathrm{eig}(\bar{\bS}^{\frac{1}{2}}\bG^H\bG\bar{\bS}^{\frac{1}{2}}+\bI)$}; \\
      \node [block_noborder] (s3) {Step 3: Calculate $\bar{\bSig}$ from (\ref{wirtp47})};\\ 
      \node [block_center4] (n4) {$\bar{\bSig}= \left(\bI-\bD^{-1}\right)^{-\frac{1}{2}}\bD^{-1}\bPhi_{\bar{\bs}}^H\bar{\bS} \bPhi_{\bar{\bs}}\bD^{-1} \left(\bI-\bD^{-1}\right)^{-\frac{1}{2}}$};\\
      \node [block_noborder] (s4) {Step 4: Calculate $\bU$ and $\bOmega$ from (\ref{wirtp50})};\\ 
      \node [block_center2] (n5) {$[\bU\,,\,\bOmega]=\mathrm{eig}(\bar{\bSig})$};\\ 
      \node [block_noborder] (s5) {Step 5: Derive $\mathbf{\Lambda}_1$ from (\ref{wirtp53})};\\ 
      \node [block_center3] (n6) {$\mathbf{\Lambda}_1=\frac{1}{2}\left(-\bI+(\bI+\frac{4}{\mu}\bOmega^{-1})^\frac{1}{2}\right)$};\\ 
      \node [block_noborder] (s6) {Step 6: Calculate $\bQ^*$ according to (\ref{wirtp54})};\\ 
      \node [block_lost] (n7) {$\bQ^*=\bar{\bS}^{\frac{1}{2}}\bPhi_{\bar{\bs}}\bD^{-1}\left(\bI-\bD^{-1}\right)^{-\frac{1}{2}}\left(\bU \mathbf{\Lambda}_1\bU^H +\bI-\bD\right) \left(\bI-\bD^{-1}\right)^{-\frac{1}{2}}\bD^{-1} \bPhi_{\bar{\bs}}^H  \bar{\bS}^{\frac{1}{2}}$};\\
      \node [block_noborder] (s7) {Step 7: Find $\mu$};\\ 
     \node [block_center2] (n8){$\mu|_{\mathrm{Tr}(\bQ^*)=P_t}$};\\
    \node [block_noborder] (s8) {Step 8: Check the validity of $\bQ^*$};\\ 
    \node[block_center2] (n9){$\bU \mathbf{\Lambda}_1\bU^H \succ\bD-\bI$};\\
    };
    \begin{scope}[every path/.style=line]
     \path (n1) -- (s1); 
     \path (s1)  -- (n2);
     \path (n2)  -- (s2);
     \path (s2)  -- (n3);
     \path (n3) -- (s3);
     \path (s3) -- (n4);
     \path (n4) -- (s4);
     \path (s4) -- (n5);
    \path (n5) -- (s5);
    \path (s5) -- (n6);
    \path (n6) -- (s6);
    \path (s6) -- (n7);
    \path (n7) -- (s7);
    \path (s7) -- (n8);
    \path (n8) -- (s8);
   \path (s8) -- (n9);
    \end{scope}
  \end{tikzpicture}
\captionof{figure}{\small Flowchart for obtaining the optimal full-rank $\bQ^*$.}
\label{fig1}
\end{center}
\end{figure}

\begin{rem}\label{rem4} By replacing $\bQ^*$ given by (\ref{wirtp54}) into (\ref{wirtp25}), the optimal input covariance matrix $\bQ^*$ given by
Theorem~\ref{thm2} attains the secrecy capacity
\begin{align}
\mathcal{C}_{sec}(P_t)&=  \log|\mathbf{\Lambda}_1|- \log|\mathbf{\Lambda}_1+\bI|+\log|\bD|-\log|\bD-\bI|   \nonumber\\
&=\log|\bI-(\mathbf{\Lambda}_1+\bI)^{-1}|+\log|\bI+(\bD-\bI)^{-1}| \; ,
\label{wirtp56}
\end{align}
where $\mathbf{\Lambda}_1$ and $\bD$ are diagonal matrices,
respectively given by (\ref{wirtp53}) and (\ref{wirtp37}). Note that
while both $\log$ terms in (\ref{wirtp56}) return non-negative values,
the first term depends on both the channels and the power $P_t$, while
the second term just depends on the channels (see Lemma \ref{lem5}).
\end{rem}

In the following we show that for the case of
$\bH^H\bH-\bG^H\bG\succeq\b0$, i.e. when at least one of the
eigenvalues of $\bH^H\bH-\bG^H\bG$ is zero, there exists an equivalent
wiretap channel with the same secrecy capacity of the original wiretap
channel. For this case, the secrecy capacity is characterized too.

\begin{thm}\label{thm3}
For the MIMO Gaussian wiretap channel defined by direct and cross
channels $\bH$ and $\bG$ respectively, with
$\bH^H\bH-\bG^H\bG\succeq\b0$, there exists an equivalent wiretap
channel with $\bH_{eq}$ and $\bG_{eq}$ such that
$\bH_{eq}^H\bH_{eq}-\bG_{eq}^H\bG_{eq}\succ\b0$, and
$\mathcal{C}_{sec}(\bH,\bG, P_t)=\mathcal{C}_{sec}(\bH_{eq},\bG_{eq},
P_t)$.
\end{thm}
\begin{proof}
Please refer to Appendix B for the proof and the characterization of 
the equivalent channels $\bH_{eq}$ and $\bG_{eq}$.
\end{proof}

Although at this time a proof is unavailable, we conjecture that the
secrecy capacity of any general MIMO wiretap channel is given by an
equivalent wiretap channel with $\bH_{eq}$ and $\bG_{eq}$ such that
$\bH_{eq}^H\bH_{eq}^H-\bG^H_{eq}\bG_{eq}\succ\b0$, and
$\mathcal{C}_{sec}(\bH,\bG, P_t)=\mathcal{C}_{sec}(\bH_{eq},\bG_{eq},
P_t)$. If one cannot find such an equivalent wiretap channel, the
secrecy capacity of the main channel is zero, i.e.,
$\mathcal{C}_{sec}(\bH,\bG, P_t)=0$.  Besides the case of
$\bH^H\bH-\bG^H\bG\succeq\b0$, the case for which the optimal input
covariance matrix is rank 1 is another case supporting this
conjecture. For the rank~1 case, it is shown in \cite{KhistiMISO} that
the optimal input covariance matrix is given by $\bQ^*=P_t\,
\bu\bu^H$, where $\bu$ is the normalized principal eigenvector
corresponding to the largest eigenvalue $\lambda_1$ of the pencil
$(\bI+P_t\bH^H\bH\;,\;\bI+P_t\bG^H\bG)$. For this case, if
$\lambda_1>1$, the secrecy capacity is $\mathcal{C}_{sec}(\bH,\bG,
P_t)=\log(\lambda_1)$, and the equivalent wiretap channel is defined
such that $\bh^H_{eq}\bh_{eq}=\bu^H\bH^H\bH\bu$ and
$\bg^H_{eq}\bg_{eq}= \bu^H\bG^H\bG\bu$.

\section{REMARKS REGARDING $\bQ^*$}

This section discusses some interesting points regarding the optimal
solution in~(\ref{wirtp54}). For the following observations, one can
assume when required that both conditions for a full-rank $\bQ^*$
given in Theorem (\ref{thm2}) are satisfied.
Let $\gamma_i$, $i=1,\cdots, n_t$, be the generalized eigenvalues of
the pencil $(\bH^H\bH\;,\;\bG^H\bG)$. Then from the definition
\cite{Charles}, $\gamma_i=\sigma_i^2$, where $\sigma_i$ is the $i$th
generalized singular value of $(\bH\;,\;\bG)$.
\begin{lem}\label{lem5}
The second term in the secrecy capacity expression  (\ref{wirtp56}) is only channel dependent and is equal to $\sum_{i=1}^{n_t}\log(\sigma^2_i)$. 
\end{lem}
\begin{proof}
From (\ref{wrtap7}) we have:
\begin{align}
\bH^H\bH \;\bar{\bS}^{\frac{1}{2}}\bPhi_{\bar{\bs}}&=\bar{\bS}^{-\frac{1}{2}}\, \bPhi_{\bar{\bs}}\;\bD  \nonumber\\
&= \bG^H\bG\; \bar{\bS}^{\frac{1}{2}}\bPhi_{\bar{\bs}}  \left(\bD-\bI\right)^{-1}\;\bD\label{wirtp57}\\
&=  \bG^H\bG\; \bar{\bS}^{\frac{1}{2}}\bPhi_{\bar{\bs}} \left(\bI+(\bD-\bI)^{-1}\right) \; , \nonumber
\end{align}
where (\ref{wirtp57}) comes from (\ref{wrtap8}). Thus, from the
definition \cite{Horn}, the generalized eigenvalue matrix of
$(\bH^H\bH\;,\;\bG^H\bG)$ is $\left(\bI+(\bD-\bI)^{-1}\right)$, which
completes the proof.
\end{proof}

Note that the definition of singular values here is slightly different
from what is given in \cite{KhistiMIMO}. Here $\sigma_i$ may be
$\infty$, while this is not the case in \cite{KhistiMIMO}. More
precisely, from (\ref{wrtap8}) for the case of $\mathrm{rank}(\bG)=n_e<n_t$, $n_t-n_e$
diagonal elements of $\bD$ are equal to one, as mentioned in Remark
\ref{remap1}. The $\sigma_i$ corresponding to $d_i=1$ tends to $\infty$.

\begin{lem}\label{lem6}
In the high SNR scenario $(P_t\rightarrow \infty)$ and for the case of
$\mathrm{rank}(\bG)=n_t$, the asymptotic form of the exact secrecy
capacity (\ref{wirtp56}) is simply given by
\begin{align} \label{wirtp58}
\mathcal{C}_{sec}&= \log|\bI+(\bD-\bI)^{-1}| =\sum_{i=1}^{n_t}\log(\sigma^2_i)
\; .
\end{align}
\end{lem}
\begin{proof}
For $P_t\rightarrow\infty$, the Lagrange parameter satisfies $\mu\rightarrow
0$, as mentioned after Theorem \ref{thm2}. Moreover, for the case of
$\mathrm{rank}(\bG)=n_t$, the matrix $\bar{\bSig}$ given by
(\ref{wirtp47}) will have finite-valued eigenvalues. Thus as $\mu\rightarrow
0$, the elements of the diagonal matrix $\mathbf{\Lambda}_1$,
given by (\ref{wirtp53}), go to $\infty$
($\lambda_{i1}\rightarrow\infty$). Consequently, the first $\log$ term
in (\ref{wirtp56}) disappears as $P_t\rightarrow\infty$.
\end{proof}

It is also interesting to consider the optimal solution in
(\ref{wirtp54}) for the case that the eavesdropper's channel is very
weak, e.g. $\bG=\b0$. For this specific case, the wiretap channel 
simplifies to a point-to-point MIMO Gaussian link, where the optimal
input covariance matrix under the average power constraint is known to
be $\bPhi_{H}\; (\frac{1}{\mu}\bI-\mathbf{\Lambda}^{-1}_H)^+
\;\bPhi_{H}^H$, and is found via the standard water-filling solution, where
unitary $\bPhi_{H}$ and diagonal $\mathbf{\Lambda}_H$ are obtained
from the eigenvalue decomposition $\bH^H\bH=\bPhi_{H}
\mathbf{\Lambda}_H\bPhi^H_{H}$.
\begin{lem}\label{lem7}
For the case of $\bG=\b0$, the optimal solution in~(\ref{wirtp54}) 
simplifies to the conventional water-filling solution, where
\begin{align} \label{wirtp59}
\bQ^*= \bPhi_{H}\; (\mu^{-1}\bI-\mathbf{\Lambda}^{-1}_H)^+ \;\bPhi_{H}^H \; .
\end{align}
\end{lem}
\begin{proof}
Using (\ref{wrtap7}) and via simple calculations, we note that for any $\bG$,  Eq. (\ref{wirtp54}) can be rewritten as 
\begin{align} \label{wirtp60}
\bQ^*=\bar{\bS}^{\frac{1}{2}}\bPhi_{\bar{\bs}}\left(\bI-\bD^{-1}\right)^{-\frac{1}{2}}\bD^{-1}\;\bU \mathbf{\Lambda}_1\bU^H\; \bD^{-1}\left(\bI-\bD^{-1}\right)^{-\frac{1}{2}} \bPhi_{\bar{\bs}}^H  \bar{\bS}^{\frac{1}{2}} - (\bH^H\bH)^{-1} \; .
\end{align}
From (\ref{wrtap8}), when $\bG\rightarrow\b0$ then
$\bD\rightarrow\bI$. Next from (\ref{wrtap7}),
$\bPhi_{\bar{\bs}}\rightarrow\bPhi_{H}$. Using these facts in
(\ref{wirtp47}), and after some straightforward calculations, we have
$\bar{\bSig}\rightarrow (\bD-\bI)^{-1}\bD^{-1} \mathbf{\Lambda}^{-1}_H
= \bOmega$, and $\bU\rightarrow\bI$ in (\ref{wirtp50}). Using these in
(\ref{wirtp60}), we have
\begin{align}\label{wirtp61}
\bQ^*&\rightarrow \bPhi_{H} (\bD-\bI)^{-1}\bD^{-1} \mathbf{\Lambda}^{-1}_H \mathbf{\Lambda}_1 \bPhi^H_{H} - (\bH^H\bH)^{-1} \nonumber\\
&\rightarrow \frac{1}{2} \bPhi_{H} (\bD-\bI)^{-1} \mathbf{\Lambda}^{-1}_H \left(-\bI+(\bI+\frac{4}{\mu}(\bD-\bI)\, \mathbf{\Lambda}_H )^\frac{1}{2}\right) \bPhi^H_{H} - (\bH^H\bH)^{-1}\nonumber\\
&\rightarrow \frac{1}{\mu}\bI  - (\bH^H\bH)^{-1} = \bPhi_{H}\; (\mu^{-1}\bI-\mathbf{\Lambda}^{-1}_H)^+ \;\bPhi_{H}^H \; , 
\end{align}
where in obtaining (\ref{wirtp61}) we used the fact that
$\bD\rightarrow\bI$ when $\bG\rightarrow\b0$.
\end{proof}

\section{Numerical Results}

\begin{figure}[t]
\centering
\includegraphics[width=3.5in,height=3.5in]{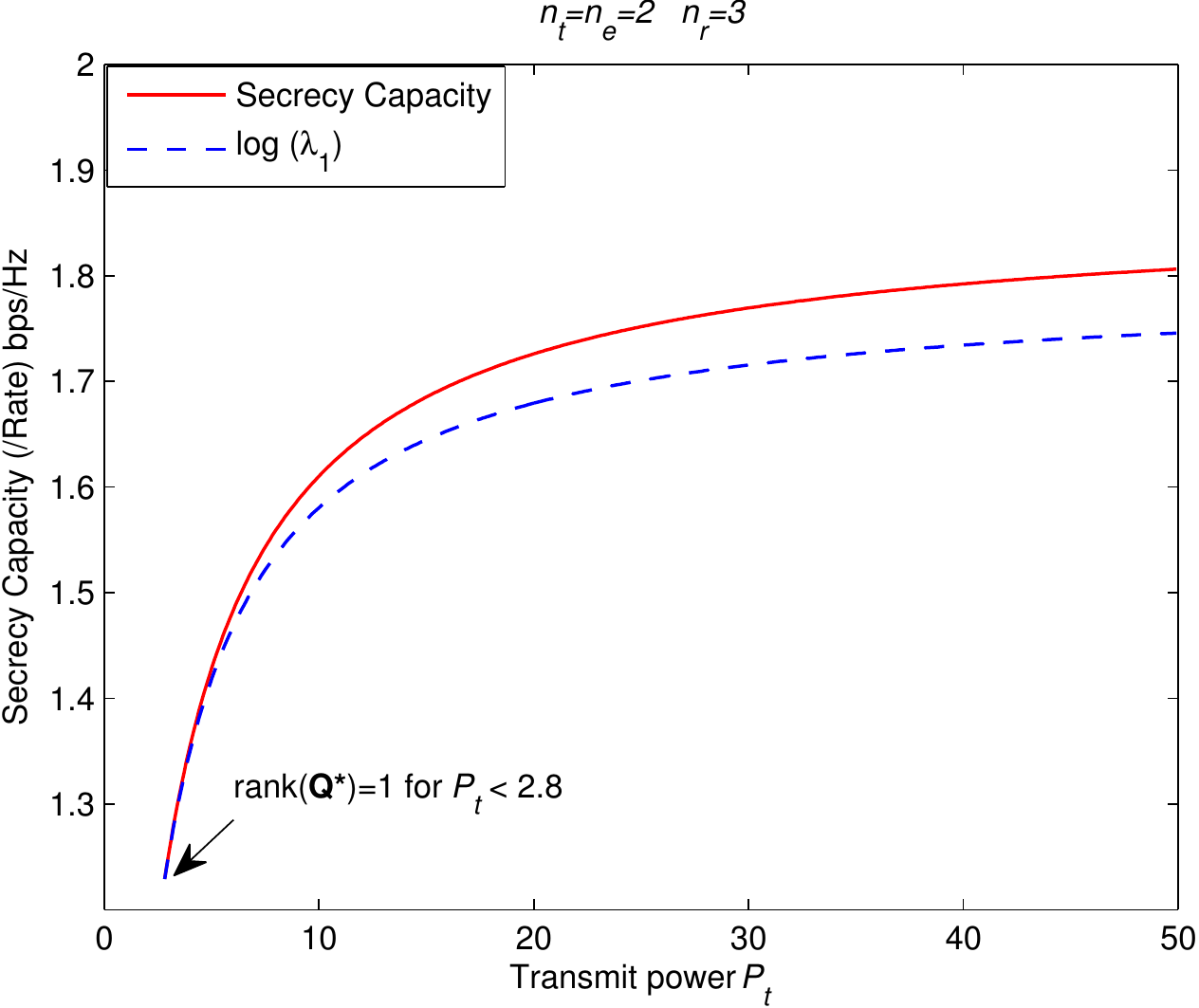}
\caption{Secrecy Capacity versus $P_t$ for $n_t=n_e=2$ and
$n_r=3$. Solid curve represents secrecy capacity and dotted curve
indicates the achievable secrecy rate using a rank-one input covariance
matrix. }
\label{wirtpexm1}
\end{figure}

In the first example, we consider a MIMO wiretap
channel with $n_t=n_e=2$, $n_t=3$ and channel matrices given by
\begin{align}
\bH=\left[\begin{array}{ccc}
   0.32 - 0.52i &  0.83 + 1.15i\\
   0.51 - 0.26i  & 0.06- 0.15i\\
  -0.11 + 0.81i  & 0.29 + 0.68i
\end{array}
\right], \label{wirtp62}\\ 
\bG=\left[\begin{array}{ccc}
   0.03 - 0.70i & -0.32 - 0.32i\\
   0.24 - 0.11i &  1.36 + 0.18i
\end{array}
\right] , \label{wirtp63}
\end{align} 
which satisfy
$\bH^H\bH-\bG^H\bG\succ\b0$. Fig.~\ref{wirtpexm1} shows the secrecy
capacity as a function of transmit power $P_t$.  For
comparison, the figure also depicts the achievable secrecy rate using
the input covariance matrix $\bQ=P_t\,\bu\bu^H$, which results to
$R_{sec}=\log \lambda_1$, as shown in
\cite{KhistiMIMO}-\cite{KhistiMISO}. Note that in this example, the
optimal $\bQ^*$ is not full-rank for $P_t<2.8$.

In Fig.~\ref{wirtpexm2} we consider another example of the case of
$\bH^H\bH-\bG^H\bG\succ\b0$, here with $n_t=n_e=3$, $n_t=4$ and 
channel matrices given by 
\begin{align}
\bH=\left[\begin{array}{ccc}
   0.89 + 0.54i & -0.06 + 0.60i &  0.48 - 1.11i\\
   0.46 &  -0.44 + 0.80i & -1.07 + 0.63i \\
   1.40 - 0.13i &  0.17 - 0.82i &  0.59 - 0.31i\\
   0.43 - 0.23i  & 0.03 + 1.35i  & 0.44 - 0.07i
\end{array}
\right], \label{wirtp64}\\ 
\bG=\left[\begin{array}{ccc}
   0.46 - 0.59i &  0.24 - 0.01i & -0.37 + 0.15i\\
   0.51 - 0.63i &  0.58 + 0.51i &  0.86 - 0.47i\\
   0.17 - 0.24i & -0.83 + 0.51i &  0.04 - 0.64i
\end{array}
\right] . \label{wirtp65}
\end{align} 
For this example, the optimal $\bQ^*$ is only full-rank
for $P_t > 0.5$.

\begin{figure}[t]
\centering
\includegraphics[width=3.5in,height=3.5in]{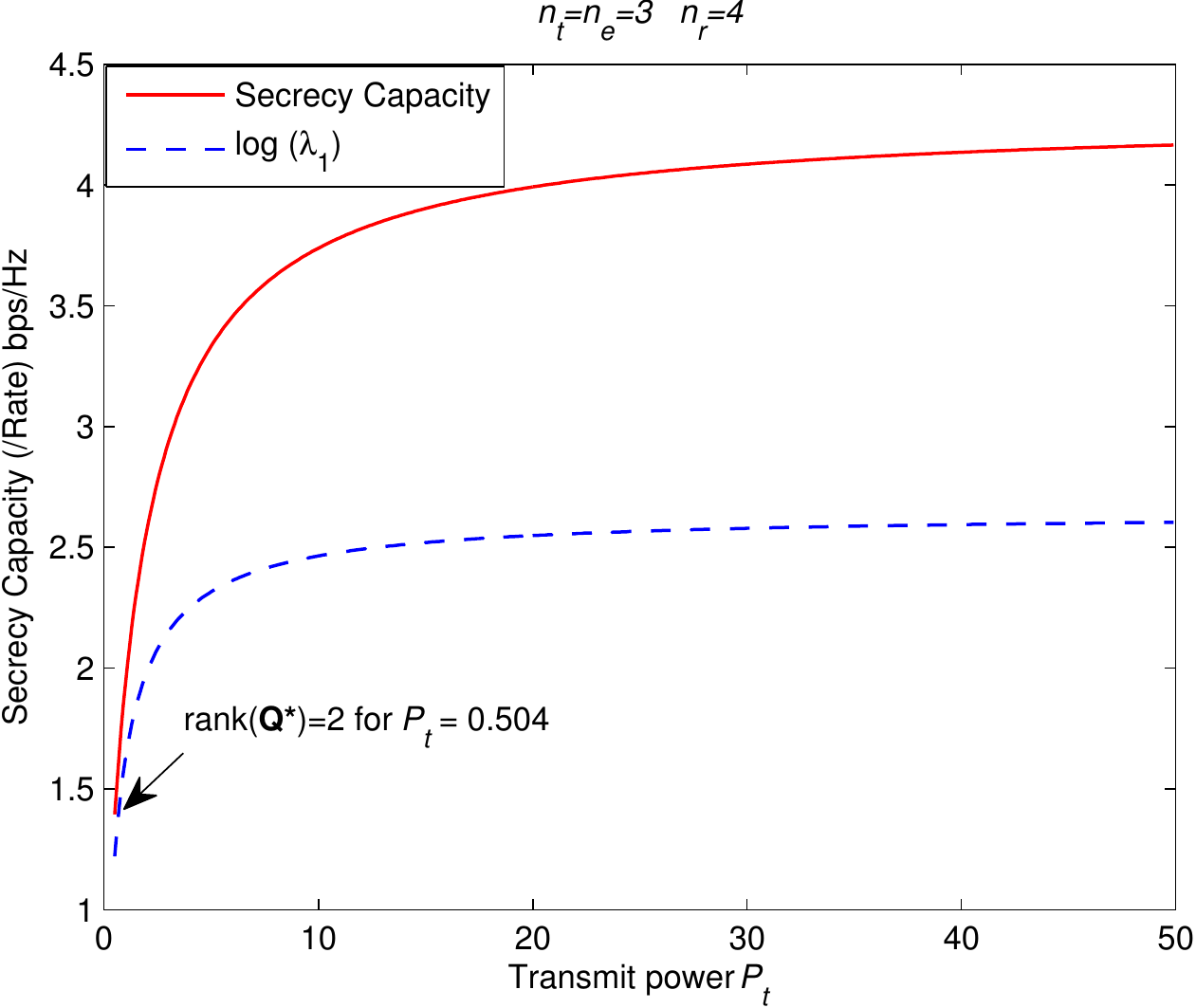}
\caption{Secrecy Capacity versus $P_t$ for $n_t=n_e=3$ and $n_r=4$. Solid curve represents secrecy capacity and dotted curve indicates the achievable secrecy rate using a rank-one input covariance matrix. }
\label{wirtpexm2}
\end{figure}
\begin{figure}[h]
\centering
\includegraphics[width=3.5in,height=3.5in]{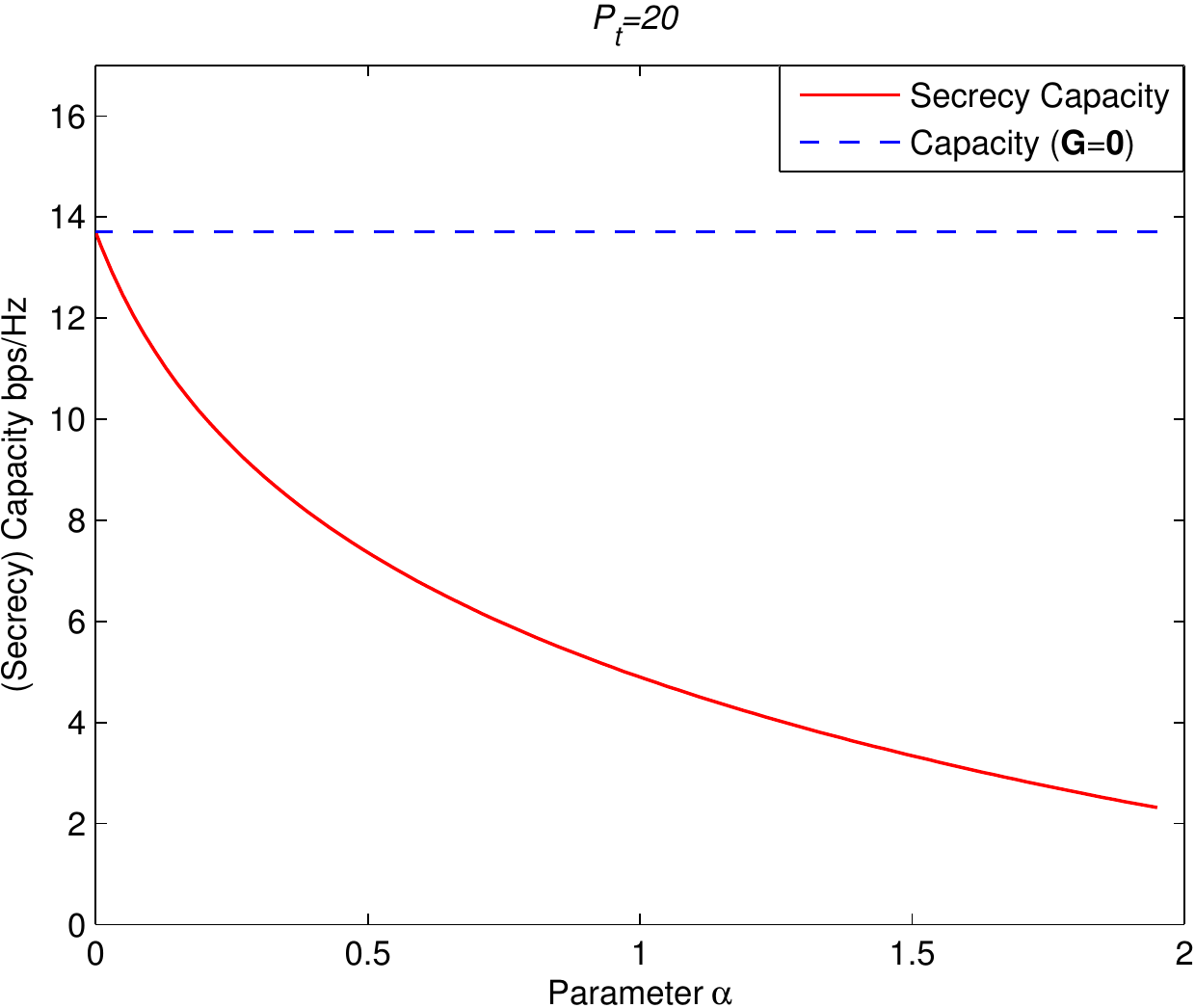}
\caption{Secrecy Capacity versus $\alpha$ for $P_t=20$. Solid curve represents secrecy capacity and dashed curve indicates the point to point capacity. }
\label{P2Pcap}
\end{figure}

Finally  in Fig.~\ref{P2Pcap}, we compare the standard point-to-point capacity without secrecy
constraints with the secrecy capacity given by (\ref{wirtp56}). In this example, $P_t=20$, direct channel $\bH$ is given by (\ref{wirtp64}) but the cross channel $\bG$ is assumed to satisfy  $\bG^H\bG=\alpha\bI$, where $\alpha$ changes from $0$ to $1.95$ (note that $\bH^H\bH-\bG^H\bG\succ\b0$ only for $\alpha\le1.95$). As predicted, the secrecy capacity achieved by the derived $\bQ^*$ in (\ref{wirtp54}) approaches the standard capacity as $\bG\rightarrow\b0$. It is interesting to note that even
for very small values of $\alpha$, the difference between the standard capacity and
secrecy capacity is considerable.

\section{Conclusion}
In this paper, we considered the rank property of the optimal input
covariance matrix under the average power constraint for a general
MIMO Gaussian wiretap channel, where each node has an arbitrary number of
antennas. We obtained necessary and sufficient constraints on
the MIMO wiretap channel parameters such that the optimal input
covariance matrix is full-rank, and we presented a method for
characterizing the resulting covariance matrix as well.

\appendices
\section{Proof of Lemma \ref{lem4}}
Define $\bar{\bS}=(\bH^H\bH-\bG^H\bG)^{-1}$ and apply the generalized
eigenvalue decomposition on the pencil
$(\bar{\bS}^{\frac{1}{2}}\bH^H\bH\bar{\bS}^{\frac{1}{2}}+\bI\;,\;\bar{\bS}^{\frac{1}{2}}\bG^H\bG\bar{\bS}^{\frac{1}{2}}+\bI)$ to obtain 
the invertible generalized eigenvector
matrix $\bar{\bC}$ and the diagonal generalized eigenvalue matrix
$\mathbf{\Lambda}_{\bar{\bs}}$ as
\begin{align}
\bar{\bC}^H\left[\bar{\bS}^{\frac{1}{2}}\bH^H\bH\bar{\bS}^{\frac{1}{2}}+\bI\right]\bar{\bC} &=\mathbf{\Lambda}_{\bar{\bs}} \label{wrtap1}\\
\bar{\bC}^H\left[\bar{\bS}^{\frac{1}{2}}\bG^H\bG\bar{\bS}^{\frac{1}{2}}+\bI\right]\bar{\bC} &=\bI\;.  \label{wrtap2}
\end{align}
By subtracting (\ref{wrtap2}) from (\ref{wrtap1}), we have
\begin{align} \label{wrtap3}
\bar{\bC}^H \,\bar{\bC} = \mathbf{\Lambda}_{\bar{\bs}}-\bI\;.
\end{align}
Note that from Lemma \ref{lem3}, we have $\mathbf{\Lambda}_{\bar{\bs}}- \bI\succ\b0$. Thus, $\bar{\bC}$ must be of the form \cite{Horn}
\begin{align} \label{wrtap4}
\bar{\bC} = \bPhi_{\bar{\bs}}\,(\mathbf{\Lambda}_{\bar{\bs}}-\bI)^{\frac{1}{2}} \;,
\end{align}
where $\bPhi_{\bar{\bs}}$ is an unknown unitary matrix. In the
following, as we continue the proof, $\bPhi_{\bar{\bs}}$ is also
characterized.

By replacing (\ref{wrtap4}) in (\ref{wrtap1}) and (\ref{wrtap2}), it is revealed that the unitary matrix $\bPhi_{\bar{\bs}}$ represents the common set of eigenvectors for the matrices $\bar{\bS}^{\frac{1}{2}}\bH^H\bH\bar{\bS}^{\frac{1}{2}}+\bI$ and $\bar{\bS}^{\frac{1}{2}}\bG^H\bG\bar{\bS}^{\frac{1}{2}}+\bI$, and thus both matrices commute.  In particular, 
\begin{align}
\bPhi_{\bar{\bs}}^H\left[\bar{\bS}^{\frac{1}{2}}\bH^H\bH\bar{\bS}^{\frac{1}{2}}+\bI\right]\bPhi_{\bar{\bs}} &=\mathbf{\Lambda}_{\bar{\bs}}(\mathbf{\Lambda}_{\bar{\bs}}-\bI)^{-1}=\bI+(\mathbf{\Lambda}_{\bar{\bs}}-\bI)^{-1}   \label{wrtap5}\\
\bPhi_{\bar{\bs}}^H\left[\bar{\bS}^{\frac{1}{2}}\bG^H\bG\bar{\bS}^{\frac{1}{2}}+\bI\right]\bPhi_{\bar{\bs}} &=(\mathbf{\Lambda}_{\bar{\bs}}-\bI)^{-1} \;. \label{wrtap6}
\end{align}
Defining $\bD=(\mathbf{\Lambda}_{\bar{\bs}}-\bI)^{-1}$, from (\ref{wrtap5})-(\ref{wrtap6}) and via straightforward computation, we have
\begin{align}
\bH^H\bH &=\bar{\bS}^{-\frac{1}{2}}\, \bPhi_{\bar{\bs}}\;\bD  
\;\bPhi_{\bar{\bs}}^H  \bar{\bS}^{-\frac{1}{2}}  \label{wrtap7}\\
\bG^H\bG &=\bar{\bS}^{-\frac{1}{2}}\, \bPhi_{\bar{\bs}} \; \left(\bD-\bI\right)\; \bPhi_{\bar{\bs}}^H   \bar{\bS}^{-\frac{1}{2}} \; , \label{wrtap8}
\end{align}
which proves (\ref{wirtp35}) and (\ref{wirtp36}). 
Substituting~(\ref{wrtap3}) in~(\ref{wrtap2}), we also have
\begin{align}
\bI&=\bar{\bC}^H\left[\bar{\bS}^{\frac{1}{2}}\bG^H\bG\bar{\bS}^{\frac{1}{2}}+\bI\right]\bar{\bC} \nonumber\\
&=  \bar{\bC}^H\bar{\bS}^{\frac{1}{2}}\bG^H\bG\bar{\bS}^{\frac{1}{2}}\bar{\bC}+ \mathbf{\Lambda}_{\bar{\bs}}-\bI  \; , \nonumber
\end{align}
or equivalently
$$ 2\bI-\mathbf{\Lambda}_{\bar{\bs}}= \bar{\bC}^H\bar{\bS}^{\frac{1}{2}}\bG^H\bG\bar{\bS}^{\frac{1}{2}}\bar{\bC}.$$

\begin{rem}\label{remap1}
Since $\bar{\bC}^H\bar{\bS}^{\frac{1}{2}}\bG^H\bG\bar{\bS}^{\frac{1}{2}}\bar{\bC}\succeq\b0$, it results that $2\bI-\mathbf{\Lambda}_{\bar{\bs}}\succeq\b0$. Equivalently, by defining $\bD=(\mathbf{\Lambda}_{\bar{\bs}}-\bI)^{-1}\succ\b0$, we have \begin{align}
\begin{split}  \label{wrtap9}
\bI-\bD^{-1}\succeq \b0 \\
\bD-\bI\succeq \b0.
\end{split}
\end{align}
\end{rem}
Note that from (\ref{wrtap8}), if $\bG^H\bG\succ\b0$, then
$\bD-\bI\succ \b0$ and vice versa. As we will observe in Theorem
\ref{thm2}, to have a full-rank optimal input covariance matrix
$\bQ^*$, having a full-rank $\bG^H\bG$ is not required. While we
assume throughout the paper and without loss of generality that the
diagonal matrix $\bD-\bI$ is invertible, for the case of rank
deficient $\bG^H\bG$ one can follow the calculations 
in this paper assuming $\epsilon>0$ for zero-diagonal elements of
$\bD-\bI$ and letting $\epsilon\downarrow 0$ at the end (see Lemma
\ref{lem7}).

\section{Proof of Theorem \ref{thm3}}
We want to obtain $\bQ_{d}^*$, the optimal input covariance matrix
that attains the secrecy capacity for the case of
$\bH^H\bH-\bG^H\bG\succeq\b0$. We note that
$\mathrm{rank}(\bH^H\bH-\bG^H\bG)=m<n_t$. Hence, from Theorems
\ref{thm1} and \ref{thm2}, $\bQ_{d}^*$ is rank-\emph{deficient}.

The right hand side of (\ref{wirtp25}) can be rewritten as
\begin{align}
R(\bQ)&=\log|\bI+\bH^H\bH\bQ|-\log|\bI+\bG^H\bG\bQ|\nonumber\\
&= \log|(\bI+\bH^H\bH\bQ) \, (\bI+\bG^H\bG\bQ)^{-1}| \nonumber\\
&= \log|\bI+(\bH^H\bH-\bG^H\bG) \, \bQ \, (\bI+\bG^H\bG\bQ)^{-1}| \; , 
\label{wrtap10}
\end{align}
where  Eq. (\ref{wrtap10}) is obtained from the matrix inversion lemma \cite{Horn} $(\bI+\bA)^{-1}=\bI-\bA(\bI+\bA)^{-1}$. 
Let the eigenvalue decomposition of $\bH^H\bH-\bG^H\bG$ to be denoted as
\begin{align}\label{wrtap11}
\bH^H\bH-\bG^H\bG = \bPsi \left[\begin{array}{ccc}
\mathbf{\Lambda}_m & \b0\\
\b0 & \b0 \end{array}\right] \bPsi^H \; , 
\end{align}
where $\mathbf{\Lambda}_m \succeq \b0$ is a diagonal matrix of size $m\times m$. Using  (\ref{wrtap11}) in   (\ref{wrtap10}), we have
\begin{align}
R(\bQ)& = \log\left|\bI+\bPsi \left[\begin{array}{ccc}
\mathbf{\Lambda}_m & \b0\\
\b0 & \b0 \end{array}\right] \bPsi^H \, \bQ \, (\bI+\bG^H\bG\bQ)^{-1}\right| \nonumber\\
& = \log\left|\bI+ \left[\begin{array}{ccc}
\mathbf{\Lambda}_m & \b0\\
\b0 & \b0 \end{array}\right] \bPsi^H \bQ\bPsi (\bI+\bPsi^H\bG^H\bG\bPsi \,\bPsi^H\bQ\bPsi)^{-1}\right| \; ,  \label{wrtap12}
\end{align}
where in obtaining (\ref{wrtap12}) we have used the facts that $\bPsi^H\bPsi=\bPsi\bPsi^H=\bI$ and $|\bI+\bA\bB|=|\bI+\bB\bA|$. 

Define $\bar{\bQ}=\bPsi^H\bQ\bPsi$ and $\bar{\bG}=\bG\bPsi$, so that
the optimization problem in (\ref{wirtp25}) can be rewritten as
$$\mathcal{C}_{sec}(P_t)= \max_{\bar{\bQ}\succeq\b0,\; \mathrm{Tr}(\bar{\bQ})=P_t} \; R(\bar{\bQ}) \; , $$
where 
\begin{align}\label{wrtap13}
R(\bar{\bQ}) = \log\left|\bI+ \left[\begin{array}{ccc}
\mathbf{\Lambda}_m & \b0\\
\b0 & \b0 \end{array}\right] \bar{\bQ}\, (\bI+\bar{\bG}^H\bar{\bG}\,\bar{\bQ})^{-1}\right|.
\end{align}
From right-hand side of (\ref{wrtap13}), we see that 
the optimal $\bar{\bQ}$ is of the form
\begin{align}\label{wrtap14}
\bar{\bQ}= \left[\begin{array}{ccc}
\bar{\bQ}_m & \b0\\
\b0 & \b0 \end{array}\right] \; , 
\end{align}
where $\bar{\bQ}_m \succeq \b0$ is of size $m\times m$. 
Write $\bar{\bG}^H\bar{\bG}$ as 
\begin{align}\label{wrtap15}
\bar{\bG}^H\bar{\bG} = \left[\begin{array}{ccc}
\bJ_1 & \bJ_2\\
\bJ_2^H & \bJ_3 \end{array}\right] 
\end{align}
where $\bJ_1$, $\bJ_2$ and $\bJ_3$ are of dimensions $m\times m$,
$m\times (n_t-m)$ and $(n_t-m)\times (n_t-m)$, respectively. By
substituting (\ref{wrtap14}) and (\ref{wrtap15}) into (\ref{wrtap13}),
we obtain
\begin{align}\label{wrtap16}
R(\bar{\bQ}) &= \log\left|\bI+ \left[\begin{array}{ccc}
\mathbf{\Lambda}_m \bar{\bQ}_m & \b0\\
\b0 & \b0 \end{array}\right] \left[\begin{array}{ccc}
\bI+\bJ_1\bar{\bQ}_m & \b0\\
\bJ_2^H\bar{\bQ}_m & \bI \end{array}\right]^{-1}\right|\nonumber\\
&= \log\left|\bI+ \left[\begin{array}{ccc}
\mathbf{\Lambda}_m \bar{\bQ}_m & \b0\\
\b0 & \b0 \end{array}\right] \left[\begin{array}{ccc}
(\bI+\bJ_1\bar{\bQ}_m)^{-1} & \b0\\
-\bJ_2^H\bar{\bQ}_m(\bI+\bJ_1\bar{\bQ}_m)^{-1} & \bI \end{array}\right]\right|\nonumber\\
&= \log\left|\bI+ \left[\begin{array}{ccc}
\mathbf{\Lambda}_m \bar{\bQ}_m(\bI+\bJ_1\bar{\bQ}_m)^{-1} & \b0\\
\b0 & \b0 \end{array}\right]\right| \nonumber\\
&= \log\left|\left[\begin{array}{ccc}
\bI+\mathbf{\Lambda}_m \bar{\bQ}_m(\bI+\bJ_1\bar{\bQ}_m)^{-1} & \b0\\
\b0 & \bI \end{array}\right]\right| \nonumber\\
&=\log\left|\bI+\mathbf{\Lambda}_m \bar{\bQ}_m(\bI+\bJ_1\bar{\bQ}_m)^{-1}\right|\nonumber\\
&=\log\left|\bI+(\mathbf{\Lambda}_m+\bJ_1) \bar{\bQ}_m\right|-\log\left|\bI+\bJ_1\bar{\bQ}_m\right|=R(\bar{\bQ}_m).
\end{align}
Using (\ref{wrtap16}), the secrecy capacity is given by
\begin{align}\label{wrtap17}
\mathcal{C}_{sec}(P_t)= \max_{\bar{\bQ}_m\succeq\b0,\; \mathrm{Tr}(\bar{\bQ}_m)=P_t} \; R(\bar{\bQ}_m) \; . 
\end{align}

Problem~(\ref{wrtap17}) shows that the secrecy capacity of a wiretap channel with $\bH^H\bH-\bG^H\bG\succeq\b0$ is equal to the secrecy capacity of an equivalent wiretap channel with  
\begin{align}
\bH_{eq}^H\bH_{eq} &=\mathbf{\Lambda}_m+\bJ_1 \label{wrtap18}\\
\bG_{eq}^H\bG_{eq} &= \bJ_1 \; , \label{wrtap19}
\end{align}
where $\mathbf{\Lambda}_m$ and $\bJ_1$ are respectively given by
(\ref{wrtap11}) and (\ref{wrtap15}).  It should also be noted that for
the equivalent channel, $\bH_{eq}^H\bH_{eq}- \bG_{eq}^H\bG_{eq} =
\mathbf{\Lambda}_m\succ \b0$. Thus, the optimal $\bar{\bQ}_m^*$ can be
computed using Theorem \ref{thm2}, as long as the equivalent channel
satisfies the second condition in Theorem \ref{thm2}. Finally, by
substituting $\bar{\bQ}_m^*$ back into (\ref{wrtap14}) we obtain
\begin{align}\label{wrtap20}
\bQ_d^*=\bPsi \left[\begin{array}{ccc}
\bar{\bQ}^*_m & \b0\\
\b0 & \b0 \end{array}\right] \bPsi^H\, , 
\end{align} 
which completes the proof.

\end{document}